\documentclass[11pt]{article}
\usepackage[a4paper, vmargin=3cm,hmargin=3.2cm,nohead]{geometry}

\usepackage[utf8]{inputenc}
\usepackage[dvipsnames]{xcolor}

\usepackage{natbib}
\usepackage{times}
\usepackage{soul}
\usepackage{url}
\usepackage[hidelinks]{hyperref}
\usepackage[utf8]{inputenc}
\usepackage[small]{caption}
\usepackage{graphicx}
\usepackage{amsmath}
\usepackage{amsthm}
\usepackage{booktabs}
\usepackage{algorithm}
\usepackage{algorithmic}
\usepackage[switch]{lineno}

\usepackage{comment}
\usepackage{pgfplots}
\pgfplotsset{compat = newest}
\usepackage{amssymb}
\usepackage{enumitem}
\usepackage{tikz}
\usepackage{nicefrac}

\definecolor{mred}		{RGB}{204,0,0}
\definecolor{mgreen}	{RGB}{0, 153, 0}
\definecolor{mblue}		{RGB}{51, 51, 255}

\usepackage{todonotes}

\usepackage{xcolor}
\pgfplotsset{compat=1.17}
\usepackage{comment}
\usepackage{pifont}

\usepackage{wrapfig}

\newtheorem{theorem}{Theorem}

\newtheorem{definition}{Definition}
\newtheorem{proposition}{Proposition}

\usepackage{bm}

\newcommand{\iit}{\textit}
\newcommand{\mcal}{M}

\newcommand{\lbr}{\left(}
\newcommand{\rbr}{\right)}

\DeclareMathOperator{\pg}{PG}

\newcommand{\opt}{\mcal^*}

\newcommand{\pga}{\pg^{\alpha}}

\newenvironment{decisionproblem}[3]{
\par\smallskip\noindent
#1: #2#3
\par\smallskip
}{}

\usepackage{bm}

\title{Beyond Stability: Improved Efficiency Guarantees for $\alpha$-Stable Matchings}
\author{Isabel Fernandez Abad\thanks{Centrum Wiskunde \& Informatica (CWI), Amsterdam, The Netherlands} \and Sophie Klumper\footnotemark[1] \and Guido Sch\"afer\thanks{Centrum Wiskunde \& Informatica (CWI) and University of Amsterdam, The Netherlands}}

\date{}

\begin{document}

\maketitle

\begin{abstract}
Stable matching mechanisms are fundamental to market design but face an inherent tension between stability and social welfare optimality. We study a natural relaxation of stability, termed \emph{$\alpha$-stability}, which models agents as willing to deviate only when the potential improvement is sufficiently large. Under $\alpha$-stability, no pair of agents can deviate and improve their valuations by more than a factor of $1/\alpha$, with $\alpha \in (0,1]$.
We provide a complete characterization of the stability--efficiency tradeoff under asymmetric valuations. This tradeoff depends on the degree of asymmetry $\mu \in (0,1]$, which bounds the ratio between agents’ valuations for any pair. 
Our results show that relaxing stability can substantially improve achievable efficiency guarantees. 
We further present a polynomial-time algorithm that computes an $\alpha$-stable matching attaining the best possible efficiency guarantee.
For $\alpha \le \mu/(\mu+1)$, our algorithm achieves $1$-efficiency; for larger $\alpha$, it computes an $\alpha$-stable matching achieving at least $(1/\alpha)\cdot \mu/(\mu+1)$ of the optimal social welfare. Remarkably, our algorithm inflates the values of an optimal matching and then applies the Gale–Shapley algorithm to the modified instance.
Finally, we show that computing an optimal $\alpha$-stable matching is NP-hard, 
even under slight relaxations of stability, i.e., for $\alpha$ close to 1.
\end{abstract}

\section{Introduction}

Matching theory studies how to assign agents to each other based on their preferences in a way that satisfies certain desirable properties.
One of the most extensively studied problems in this field is the \emph{stable marriage problem} \citep{gale1962college,Gusfield1989Stable,Knuth1997,Roth1990Twosided}, which involves two sets of agents $L$ and $R$ of size $n \ge 1$ each, where each agent ranks all members of the opposite set in a strict order of preference.
A matching is said to be \emph{stable} if there is no pair of agents who would both prefer to be matched with each other over their current matches. 
Stable matchings have been a cornerstone of market design with wide-ranging applications in assigning students to schools, medical graduates to residency programs, employees to projects, or students to courses and seminars.

In this work, we study a more general variant of the classical stable matching problem, where agents have \emph{cardinal preferences} over the members of the opposite group, rather than ordinal ones. Such cardinal preferences do not only capture who is preferred over whom, but also quantify by how much. Additionally, we consider cardinal preferences that can encode \emph{indifferences} and \emph{incompatibilities} among the agents. 
Each agent's cardinal preferences are defined by numerical values (called \emph{valuations}) specifying how much the agent would like to be matched to another agent. 
More specifically, we use $v_{ij}$ to denote the (positive) valuation that an agent $i \in L$ assigns to $j \in R$, and $w_{ji}$ to denote the (positive) valuation that agent $j \in R$ assigns to $i \in L$. 
This variant of the stable matching problem is also known as the \emph{stable matching problem with cardinal preferences, ties allowed, and incomplete preference lists} (or, \emph{SMCTI} for short). 
Stable matchings are guaranteed to exist for SMCTI (see, e.g., \citep{Gusfield1989Stable}). 

We are interested in the tradeoff between stability and the social welfare efficiency. Given a matching $M$, the \emph{social welfare} refers to the total valuations of all matched pairs, i.e., $SW(M) = \sum_{(i,j) \in M} (v_{ij} + w_{ji})$. A stable matching $M$ is said to be \emph{$\gamma$-efficient} with $\gamma \in (0,1]$ if it recovers a $\gamma$-fraction of the optimal social welfare (i.e., the maximum social welfare over \emph{all} matchings). 
Ideally, we would like to ensure that a stable matching provides a strong efficiency guarantee. 
However, imposing stability inevitably leads to a significant loss in social welfare efficiency (see the example below).
As it turns out, this loss depends on the \emph{degree of asymmetry} in the agents’ valuations. 
We capture this by introducing an asymmetry parameter $\mu \in (0,1]$: the valuations are \emph{$\mu$-bounded} if $\nicefrac{v_{ij}}{w_{ji}} \in [\mu, \nicefrac{1}{\mu}]$ for every pair $(i,j)$ with $v_{ij} > 0$ and $w_{ji} > 0$. 
Note that the valuations are symmetric for $\mu = 1$, and become increasingly asymmetric as $\mu$ decreases.

\begin{figure}[t]
    \centering
    \begin{tikzpicture}[scale=2]
        \node[label=left:$i_1$] (a) at (0,1) {};
        \node[label=left:$i_2$] (b) at (0,0) {};
        \node[label=right:$j_1$] (c) at (1,1) {};
        \node[label=right:$j_2$] (d) at (1,0) {};
        
        \filldraw[black] (0,0) circle (1pt){};
        \filldraw[black] (1,0) circle (1pt){};
        \filldraw[black] (0,1) circle (1pt){};
        \filldraw[black] (1,1) circle (1pt){};

        \draw[thick] (a) -- (c) 
        node[pos=0.1, above] {\small $\frac{1}{2\mu}$} 
        node[pos=0.9, above] {\small $\frac{1}{2}$};
        \draw[thick] (b) -- (d) 
        node[pos=0.1, below] {\small $\frac{1}{2}$} 
        node[pos=0.9, below] {\small $\frac{1}{2\mu}$};
        \draw[thick] (b) -- (c) 
        node[pos=0.25, left] {\small $\frac{1+\epsilon}{2}$} 
        node[pos=0.75, right] {\small $\ \frac{1+\epsilon}{2}$};
    \end{tikzpicture}
    \caption{Example with $n=2$ agents and valuations $(v_{ij})$ and $(w_{ji})$ shown as left and right edge labels, respectively.} 
    \label{fig:SQM_mu}
\end{figure}
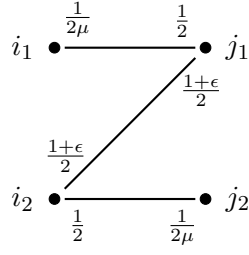

Consider the example in Figure~\ref{fig:SQM_mu} and note that the valuations are $\mu$-bounded. 
It is easy to verify that, for $\epsilon > 0$, the instance has 
a unique stable matching $M = \{(i_2, j_1)\}$. On the other hand, the social welfare optimal matching is $M^* = \{(i_1, j_1), (i_2, j_2)\}$. In terms of social welfare, we have $SW(M) = 1+\epsilon$ and $SW(M^*) = \nicefrac{(\mu+1)}{\mu}$.
Thus, the efficiency of the (unique) stable matching is $\nicefrac{(1+\epsilon)\mu}{(\mu+1)}$. 
As a consequence, for this instance we cannot guarantee that any stable matching recovers (strictly) more than a fraction of $\nicefrac{\mu}{(\mu+1)}$ of the optimal social welfare. 
In particular, the efficiency of a stable matching becomes unbounded as $\mu$ approaches $0$, or, said differently, as the valuations become increasingly asymmetric.
In fact, this tradeoff has been observed before by \citet{anshelevich2013anarchy} for the two limiting cases $\mu = 1$ and $\mu \rightarrow 0$; but, here, we get a smooth interpolation for any $\mu$ between these extreme points---which we show is tight also (see our contributions below).

A further drawback of the classical notion of stability is that it assumes agents will always deviate, even if the respective gain is vanishingly small.
This is an overly strong assumption in settings where agents may be reluctant to deviate.
As an example, consider the situation of deciding on money transfers between different accounts, where a commission discourages moves unless the gain justifies the expense. 
Another example are housing markets, where the decision to move to a better home must also account for the financial and personal costs of relocating.
In this context, an approach that has been pursued in the literature is to relax the stability notion by incorporating a fixed, additive \emph{switching cost} \citep{anshelevich2013anarchy}.
While an additive switching cost may be appropriate in some scenarios, in others a cost that scales proportionally with the valuations may be more realistic.

In this work, we build on a relaxed stability notion, termed \emph{$\alpha$-stability}, introduced by \citet{anshelevich2013anarchy}. 
The key idea is to model agents as willing to deviate only when the improvement is sufficiently large, with the required threshold being determined by a parameter $\alpha$.
Formally, a matching $\mcal$ is \emph{$\alpha$-stable} with $\alpha \in (0,1]$ if there is no pair $(i,j)$ such that $v_{i\mcal(i)}<\alpha v_{ij}$ and $w_{j\mcal(j)}<\alpha w_{ji}$, where $\mcal(i)$ and $\mcal(j)$ denote the matches of $i$ and $j$ under $\mcal$, respectively. 
As $\alpha$ decreases, agents require a larger improvement to switch, making the stability condition less restrictive.

As mentioned above, $\alpha$-stability was first introduced in \citep{anshelevich2013anarchy}, where they restrict attention to the case of symmetric valuations (i.e., $\mu = 1$).
To the best of our knowledge, we are the first to study the more complex case of asymmetric valuations that are $\mu$-bounded. 
The idea of studying relaxation of fundamental solution concepts has also been explored in adjacent fields, such as \emph{approximate Nash equilibria} by \citet{daskalakis2009note} and relaxed fairness notions such as \emph{envy-freeness up to one good (EF1)} by \citet{budish2011combinatorial}.

\begin{figure}[t]
    \centering
    \begin{tikzpicture}[scale=0.85]
      \begin{axis}[
        width=8cm,
        height=6cm,
        axis x line = bottom,
        axis y line = left,
        xtick = {0, 0.5, 1},
        ytick = {0, 0.5, 1},
        domain=0:1,
        samples=200,
        xlabel={$\alpha$-stability},
        ylabel={$\gamma$-efficiency},
        xmin=0, xmax=1.05,
        ymin=0, ymax=1.1,  
        clip=false,
      ]

        \draw node[left] at (0, 1/3) {$\frac{\mu}{\mu+1}$};
        \draw[dotted] (-0.01, 1/3) -- (1.01, 1/3);

        \draw node[below] at (1/3, 0) {$\frac{\mu}{\mu+1}$};        
        \draw[dotted] (1/3, -0.01) -- (1/3, 1);

        \addplot[mred, very thick, domain=0:1/3] {1};
        \node[below] at (0.7,0.75) {$\frac{1}{\alpha}\cdot \frac{\mu}{\mu+1}$};
        \addplot[mred, very thick, domain=0.332:1] {1/x * 1/3};

        \node[below] at (0.7, 0.22) {$\alpha \cdot \frac{\mu}{\mu+1}$};
        \addplot[mgreen, very thick, domain=0:1] 
          {x * 1/3};

      \end{axis}
    \end{tikzpicture}
    \caption{Tradeoff between $\alpha$-stability and $\gamma$-efficiency for $\mu = \frac{1}{2}$. Lower curve (green): lower bound on the efficiency of any $\alpha$-stable matching. Upper curve (red): upper bound on the best-possible efficiency guarantee of $\alpha$-stable matchings.}
    \label{fig:stmin-plot}
\end{figure}
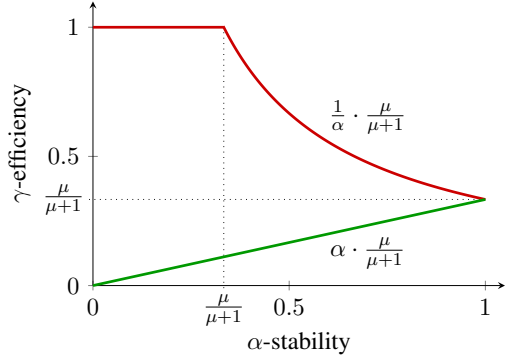

\subsection{Our Contributions} 

We study how the relaxed notion of $\alpha$-stability influences social welfare efficiency under asymmetric ($\mu$-bounded) valuations, which we term the \emph{stability--efficiency tradeoff}.

Our main results are as follows:
\begin{enumerate}\itemsep10pt
    \item \iit{Characterization of Stability--Efficiency Tradoff:}
    We provide a complete characterization of the stability--efficiency tradeoff (Section 3); see Figure~\ref{fig:stmin-plot} for an illustration. 
    We prove that any $\alpha$-stable matching is at least $(\alpha \cdot \nicefrac{\mu}{(\mu+1)})$-efficient (Proposition \ref{prop:bounds2}). Thus, the worst-case efficiency of $\alpha$-stable matchings becomes worse as $\alpha$ decreases.
    On the other hand, we prove an upper bound of $\nicefrac{1}{\alpha}\cdot \nicefrac{\mu}{(\mu+1)}$ on the best-case efficiency of $\alpha$-stable matchings.     
    In particular, for smaller values of $\alpha$, $\alpha$-stable matchings can achieve a (significantly) better efficiency guarantee than stable matchings, with the magnitude of improvement depending on the ratio $\nicefrac{\mu}{(\mu+1)}$.
    Also, for $\alpha \le \nicefrac{\mu}{(\mu+1)}$, any optimal matching is $\alpha$-stable, implying that $\alpha$-stable matchings can be $1$-efficient. 
    At the same time, our result shows that it is impossible to surpass the bound of $\min(1, \nicefrac{1}{\alpha}\cdot \nicefrac{\mu}{(\mu+1)})$.     
    This provides a complete characterization of the stability--efficiency tradeoff for $\alpha$-stable matchings under $\mu$-bounded valuations; previously, only the case $\mu = 1$ was understood \citep{anshelevich2013anarchy}.

    \item \iit{Computing $\alpha$-Stable Matchings:} 
    We present a polynomial-time algorithm, termed \emph{$\alpha$-Boosting}, that computes an $\alpha$-stable matching achieving the best-possible worst-case efficiency guarantee (Section 4, Theorem \ref{th:boostingAlg}). 
    For $\alpha \leq \nicefrac{\mu}{(\mu + 1)}$, the algorithm simply returns a social welfare optimal matching, which is $\alpha$-stable and achieves a 1-approximation.     
    For $\alpha > \nicefrac{\mu}{(\mu + 1)}$, $\alpha$-Boosting uses a more sophisticated approach inspired by \citep{Colinibaldeschi2024Trust}:
    it first computes a social welfare optimal matching, inflates the values of its edges by a factor $\nicefrac{1}{\alpha}$, and then applies the Gale--Shapley algorithm to the modified instance. 
    Remarkably, we can show that this algorithm has an approximation guarantee of $\nicefrac{1}{\alpha} \cdot \nicefrac{\mu}{(\mu+1)}$, which is best possible. 
    
    \item \iit{Hardness of Optimal $\alpha$-Stable Matchings:}
    We consider the problem of computing an $\alpha$-stable matching of maximum social welfare (Section 5).
    We derive two integer linear programming formulations that characterize the set of $\alpha$-stable matchings. 
    For the stable matching problem with strict preference orders, \citet{VandeVate1989} and \citet{Rothblum1992} showed that the corresponding linear programming relaxations are integral. However, such a result seems elusive for $\alpha$-SMCTI. We prove that for any $\alpha \in (\nicefrac{n-1}{n}, 1)$, deciding whether an instance of $\alpha$-SMCTI admits an $\alpha$-stable matching achieving at least a given social welfare is NP-complete (Theorem \ref{thm:maxweight_NPhard}). 
    In particular, this implies that computing an optimal $\alpha$-stable matching is NP-hard for this range of $\alpha$.
    
\end{enumerate}

\subsection{Related Work} 

The stable matching problem was first introduced by \citet{gale1962college}, together with an algorithm that computes a stable matching which is optimal in terms of social welfare for the proposing side. Their model assumes that (i) preference lists are complete, and (ii) there are no ties in the preferences. \citet{Gale1985Remarks} considered the case of incomplete preferences and studied how changing the length on an agent's preference list affects their welfare in the final matching.

Later work has explored generalizations of the problem that drop both of the assumptions described above. \citet{Iwama1999Stable} showed that when both assumptions are removed, the problem of determining whether a (perfect) stable matching exists is NP-hard. When only (ii) is dropped, \citet{Manlove2002} proved that finding an optimal stable matching is NP-hard, and \citet{erdil2017two} showed that the matching resulting from the Gale--Shapley algorithm may not be optimal for either side. More recently, \citet{Haas2021} proposed heuristic methods to improve the social welfare of these matchings when ties occur.

A different way to generalize the stable marriage problem is to consider weighted preferences. \citet{anshelevich2013anarchy} do this and define the (multiplicative) notion of $\alpha$-stability, which we consider in this paper.  
They consider three scenarios: vertex-labeled graphs, symmetric edge-labeled graphs, and asymmetric edge-labeled graphs.\footnote{Our $\mu$ parameter interpolates between symmetric edge-labeled ($\mu=1$) and asymmetric edge-labeled graphs ($\mu\rightarrow 0$).} In the latter case, they showed that the social welfare of an $\alpha$-stable matching can be arbitrarily worse than that of an optimal matching, but did not explore this setting further. For symmetric edge-labeled graphs, however, they proved that there always exists an $\alpha$-stable matching with social welfare at least $\nicefrac{1}{2\alpha}$ times that of an optimal matching.\footnote{Note that in \citet{anshelevich2013anarchy}, the approximation ratios are presented inversely compared to our notation.} They also provided an algorithm that finds such a matching in $O(n^2)$ time.

$\alpha$-Stability was also analyzed in the context of fractional matchings by \citet{Caragiannis2021Stable}. They propose a polynomial time algorithm that computes an $\alpha$-stable fractional matching which guarantees an $\alpha$-approximation. They also prove that the problem of finding an optimal $\alpha$-stable fractional matching is hard to approximate. Different notions of stability have also been considered in the literature (e.g., \citep{chen2021matchings,echenique2025stable,Lin2024Stable,Pini2013Stability,Troyan2020Essentially}).

The linear programming approach to the stable matching problem has been widely studied. \citet{VandeVate1989} was the first to show that the polytope defined by the stable matching constraints has integral extreme points, thereby proving that an optimal stable matching can be found in polynomial time. \citet{Rothblum1992} later expanded on these results and offered simplified proofs. However, they do not easily generalize to the $\alpha$-stability setting. On the other hand, \citet{Manlove2002} considered a series of harder variants of the problem, which they showed to be NP-hard.

Finally, the stable marriage problem can be seen as a special case of the Generalized Assignment Problem for unit sizes and capacities. \citet{Colinibaldeschi2024Trust} study different variants of this problem when one group is strategic in a learning-augmented environment.
In particular, they introduce a mechanism, called \textit{Boost}, which increases the values of edges in the predicted matching such that these edges are more likely to be matched during a modified run of the Gale--Shapley algorithm. Our $\alpha$-Boosting algorithm builds upon this procedure.

\section{Preliminaries}

\paragraph{Stable Matching with Cardinal Preferences.~}
An instance of the \emph{stable matching problem with cardinal preferences} is defined as follows: We are given two sets $L$ and $R$, each consisting of $n$ agents.\footnote{It is w.l.o.g. that $|L|=|R|$, i.e., an instance with $|L| > |R|$ can be reduced to an instance with $|L|=|R'|$ by adding agents to $R$ and defining that $\forall j \in R' \setminus R$ it holds that $v_{ij} = w_{ji} = 0$, $\forall i \in L$.} 
We use $i$ and $j$ to refer to agents from the sets $L$ and $R$, respectively. 
Each agent has cardinal preferences (allowing for indifferences and incompatibilities) over the agents in the other set. 
More specifically, each agent $i \in L$ has a non-negative \emph{valuation $v_{ij}$} for every $j \in R$, and 
each agent $j \in R$ has a non-negative \emph{valuation $w_{ji}$} for every $i \in L$.
The valuations of an agent comprise their \iit{preference list}. 
The interpretation of the valuations $(v_{ij})$ for agent $i \in L$ is as follows: agent $i$ prefers to be matched to $j_1$ rather than $j_2$ if $v_{ij_1} > v_{ij_2}$, $i$ is indifferent between $j_1$ and $j_2$ if $v_{ij_1} = v_{ij_2}$, and $i$ does not want to be matched to $j$ if $v_{ij} = 0$. 
The interpretation of the valuations $(w_{ji})$ for an agent $j \in R$ follows analogously. 
We define $(i,j)$ as a \textit{pair} if $i \in L$ and $j \in R$.
We say that $(i,j)$ is a \iit{compatible} pair if $v_{ij}\neq 0$ and $w_{ji}\neq 0$. 
The \iit{preference graph} $\pg(L\cup R, (v_{ij}), (w_{ji}))$ (or simply $\pg$, when clear from context) is the bipartite graph that represents these relational preferences: the vertex set is $L\cup R$, the edge set consists of all compatible pairs, i.e.,
$
E(\pg) =\{(i,j)\in L\times R:v_{ij}\neq 0 \text{ and } w_{ji}\neq 0\},
$ 
and each edge $(i,j) \in E(\pg)$ is associated with two (positive) weights $v_{ij}$ and $w_{ji}$.
We assume w.l.o.g. that the preference graph $\pg$ is non-empty.

\paragraph{$\alpha$-Stable Matchings.~}

Given a preference graph $\pg$, a \iit{matching} $\mcal_{\pg}$ (or simply $\mcal$, when clear from context) for $\pg$ is a set of edges such that each vertex is incident to at most one edge. 
Unless specified otherwise, a matching refers to a non-empty matching. 
We will abuse notation and also denote by $\mcal$ the mapping $L\cup R \rightarrow L \cup R \cup \emptyset$ such that $\mcal(i)=j$ and $\mcal(j)=i$ whenever $(i,j)\in \mcal$. Otherwise, if there is no $j \in R$ such that $(i,j) \in \mcal$, we set $\mcal(i)= \emptyset$, in which case we will say that $i$ is \iit{lonely} under $\mcal$ and define $v_{i\mcal(i)}=0$. For $j \in R$, $\mcal(j)=\emptyset$ and $w_{j \mcal(j)}$ are defined similarly. 
We use $\text{Mat}(\pg)$ to refer to the set of matchings for $\pg$. A matching $\mcal$ is said to be \iit{(inclusion-wise) maximal} if there is no edge $e\in E(\pg) \setminus M$ such that $\mcal\cup \{e\}$ is also a matching.
A matching is said to be \iit{stable} if no pair strictly prefers each other to their current matches if any, or to being lonely.\footnote{This is sometimes referred to as \iit{weak stability}.} 
In this paper we consider a relaxed stability notion, called \emph{$\alpha$-stability} (see also \citep{anshelevich2013anarchy}).

\begin{definition}[$\alpha$-Stability]
Let $\mcal$ be a matching for a preference graph $\pg$, and let $\alpha\in (0,1]$.
A pair $(i,j) \in E(\pg)$ is \emph{$\alpha$-blocking} if $v_{i\mcal(i)} < \alpha v_{ij}$ and $w_{j\mcal(j)}< \alpha w_{ji}$. 
$M$ is \emph{$\alpha$-stable} if there is no $\alpha$-blocking pair.
\footnote{We impose $\alpha > 0$ as every matching is trivially $0$-stable.}
\end{definition}

Note that the original stability notion is equivalent to $1$-stability.
By definition, an $\alpha$-stable matching $\mcal$ is also $\alpha'$-stable for every $\alpha' \le \alpha$.
In particular, a stable matching is $\alpha$-stable for every $\alpha\in(0,1]$.
This observation also implies that, for any preference graph $\pg$, an $\alpha$-stable matching is guaranteed to exist for every $\alpha \in (0, 1]$ as a stable matching always exists (see \citep{Gusfield1989Stable}).
Observe that any $\alpha$-stable matching must be (inclusion-wise) maximal.

\paragraph{Benchmark.~}

We are interested in how the quality of $\alpha$-stable matchings changes depending on $\alpha$. 
To this end, we define the \iit{social welfare} of a matching $\mcal$ as 
$SW(\mcal)=\sum_{(i,j)\in \mcal}(v_{ij}+w_{ji})$.
That is, the social welfare $SW(\mcal)$ accounts for the sum of the valuations of all pairs in the matching $\mcal$. 
Given a preference graph $\pg$, an \emph{optimal matching} is a matching of maximum social welfare and will be denoted by $\mcal^*$, i.e., $\opt \in \text{Mat}(\pg)$ and $SW(\opt) \ge SW(\mcal)$ for every matching $\mcal \in \text{Mat}(\pg)$.
Throughout this work, we use the optimal social welfare as a benchmark when evaluating the quality of a matching.

\begin{definition}[$\gamma$-Efficiency]
Let $\mcal$ be a matching and $M^*$ be an optimal matching for a preference graph $\pg$. $M$ is said to be \emph{$\gamma$-efficient} with $\gamma \in (0,1]$ if $SW(\mcal)=\gamma SW(\opt)$. 
\end{definition}

\paragraph{$\mu$-Bounded Valuations.~}

As it turns out, the efficiency of $\alpha$-stable matchings does not only depend on the stability parameter $\alpha$, but also on the ``degree of asymmetry'' of the underlying valuations. The following notion formalizes this. 
Given a preference graph $\pg$, the valuations of the agents are \emph{$\mu$-bounded} with $\mu \in (0,1]$ if
$\mu w_{ji}\leq v_{ij}\leq w_{ji}/\mu$ for all $(i,j)\in E(\pg)$.
Intuitively, for $\mu = 1$ the valuations are symmetric (i.e., $v_{ij} = w_{ji}$ for all $(i,j) \in E(\pg)$),
and they become increasingly asymmetric as $\mu$ decreases. 

\paragraph{Computing $\alpha$-Stable Matchings.~}

We refer to our stable matching problem introduced above as \emph{$\alpha$-SMCTI} for short, which stands for \emph{$\alpha$-Stable Matching with Cardinal preferences, Ties allowed, and Incomplete preference lists}.
An instance of $\alpha$-SMCTI is given by a preference graph $\pg$ and the parameter $\alpha$. Given the preference graph $\pg$, we let $\mu$ refer to the (largest) $\mu$ parameter for which the valuations are $\mu$-bounded (as defined above). 
We are interested in deriving algorithms that compute $\alpha$-stable matchings with good efficiency guarantees. 
In this context, a \emph{$\rho$-approximation algorithm} with $\rho \in (0,1]$ is an algorithm that, for any given instance of $\alpha$-SMCTI, computes in polynomial time an $\alpha$-stable matching $\mcal$ for $\pg$ that is at least $\rho$-efficient, i.e., $SW(\mcal) \ge \rho \cdot SW(\mcal^*)$; $\rho$ is also called the \emph{approximation guarantee}. 
We also consider the problem of computing an \emph{optimal $\alpha$-stable matching}, i.e., an $\alpha$-stable matching of maximum social welfare among all $\alpha$-stable matchings. 

\section{Stability--Efficiency Tradeoff} 

Given an instance of $\alpha$-SMCTI, our ultimate goal is to compute an $\alpha$-stable matching with a good (large) efficiency.
We begin by analyzing the worst-case and best-case efficiency guarantees achievable by $\alpha$-stable matchings.

\begin{proposition}\label{prop:bounds1}
Consider an instance of $\alpha$-SMCTI with preference graph $\pg$ and $\mu$-bounded valuations. Then
\begin{enumerate}[label={$(\mathrm{\alph*})$}]
    \item if $\alpha\in(0,\frac{\mu}{\mu+1}]$, there is always a 1-efficient $\alpha$-stable matching, and
    \item if $\alpha\in (\frac{\mu}{\mu+1}, 1]$, then for every $\delta>0$ there exists a preference graph $\pg$ such that there is no $(\frac{1}{\alpha}\cdot \frac{\mu}{\mu+1}+\delta)$-efficient $\alpha$-stable matching.
\end{enumerate}
\end{proposition}

\begin{proof}
Let $\alpha\in(0,\frac{\mu}{\mu+1}]$. We will show that any optimal matching is $\alpha$-stable. 
Towards a contradiction, let $\opt$ be an optimal matching that is not $\alpha$-stable for some $\alpha\in (0, \frac{\mu}{\mu+1}]$. 
Then there exists an $\alpha$-blocking pair $(i,j)$, i.e., $v_{i\opt(i)}<\alpha v_{ij}$ and $w_{j\opt(j)}<\alpha w_{ji}$.
As the valuations are $\mu$-bounded we obtain
\[ v_{i\opt(i)}+w_{\opt(i)i} \leq \tfrac{\mu + 1}{\mu} v_{i\opt(i)}<\alpha\cdot\tfrac{\mu+1}{\mu}v_{ij} \]
and
\[ v_{\opt(j)j}+w_{j\opt(j)}\leq \tfrac{\mu + 1}{\mu} w_{j\opt(j)}<\alpha\cdot\tfrac{\mu+1}{\mu}w_{ji}.\]
As $\alpha\cdot\frac{\mu+1}{\mu}\leq 1$, adding the above equations leads to
\begin{equation} \label{eq:comparison-values-1}
v_{i\opt(i)}+w_{\opt(i)i}+ v_{\opt(j)j}+w_{j\opt(j)} <
v_{ij}+w_{ji}.
\end{equation}
The matching 
$\hat{\mcal}:=\opt \setminus \{(i,\opt(i)),(\opt(j),j)\} \cup \{(i,j)\}$
therefore satisfies $SW(\hat{\mcal})>SW(\opt)$, contradicting optimality of $\opt$. Observe that the argument also holds if either $i$ or $j$ is lonely under $\opt$, as \eqref{eq:comparison-values-1} remains true, and if $i$ and $j$ are both lonely under $\opt$ then this immediately contradicts optimality of $\opt$.

Now let $\alpha\in (\frac{\mu}{\mu+1},1]$. 
For each value of $\delta>0$, we will construct a preference graph such that there is no $(\frac{1}{\alpha}\cdot\frac{\mu}{\mu+1}+\delta)$-efficient $\alpha$-stable matching. Consider the instance shown in Figure \ref{fig:SQM_mu}, where $\epsilon=\frac{1}{\alpha}-1+\epsilon'$ for some  $0<\epsilon'<(1+\frac{1}{\mu})\delta$ such that $\epsilon\in (\frac{1}{\alpha}-1,\frac{1}{\mu})$.\footnote{This interval is well defined as $\frac{1}{\alpha}-1<\frac{1}{\mu}$ since $\alpha>\frac{\mu}{\mu+1}$.} 
The only $\alpha$-stable matching is given by $\mcal=\{(i_2,j_1)\}$, while the optimal matching is $\opt=\{(i_1,j_1), (i_2,j_2)\}$, leading to\footnote{Note that $SW(\opt) > 0$, as $E(\pg) \neq \emptyset$.}
\[ 
\frac{SW(\mcal)}{SW(\opt)}=\frac{1+\epsilon}{1 + \frac{1}{\mu}}=\frac{1}{\alpha}\cdot \frac{\mu}{\mu+1}+\frac{\epsilon'}{1+\frac{1}{\mu}}<\frac{1}{\alpha}\cdot \frac{\mu}{\mu+1}+\delta.\]
\end{proof}

\begin{proposition}\label{prop:bounds2}
Given an instance of $\alpha$-SMCTI with preference graph $\pg$ and $\mu$-bounded valuations, any $\alpha$-stable matching is at least 
$(\alpha\cdot \frac{\mu}{\mu+1})$-efficient.
\end{proposition}

\begin{proof}
Let $\mcal$ be an $\alpha$-stable matching, and fix an optimal matching $\opt$. Due to $\alpha$-stability, for every $(i,j)\in \opt$, it holds that $v_{i\mcal(i)}\geq \alpha v_{ij}$ or $w_{j\mcal(j)}\geq \alpha w_{ji}$. Thus we can write $\opt= E_1\cup E_2$, where $E_1 := \{(i,j)\in \opt \: : \: v_{i\mcal(i)}\geq \alpha v_{ij}\}$ and $E_2 := \{(i,j)\in \opt \: : \: w_{j\mcal(j)}\geq \alpha w_{ji}\}.$
In particular,
\begin{equation}\label{eq:sum_e1e2alpha}
SW(\opt)\leq \sum_{(i,j)\in E_1}(v_{ij}+w_{ji})+\sum_{(i,j)\in E_2}(v_{ij}+w_{ji}).
\end{equation}
Consider the first summation in \eqref{eq:sum_e1e2alpha}. 
It holds that $w_{ji} \leq \frac{1}{\mu} v_{ij} \le \frac{1}{\alpha \mu}v_{i\mcal(i)}$, as the valuations are $\mu$-bounded and by definition of $E_1$. 
Therefore,
\begin{align}\label{eq:sum_e1alpha}
\sum_{(i,j)\in E_1}( v_{ij}+w_{ji})& 
\leq \sum_{(i,j)\in \mcal}\tfrac{1}{\alpha}\lbr 1 + \tfrac{1}{\mu} \rbr v_{ij}.
\end{align}
We can upper bound the second summation in \eqref{eq:sum_e1e2alpha} analogously, and therefore upper bound \eqref{eq:sum_e1e2alpha} by $SW(\opt)\leq \tfrac{1}{\alpha}( 1+\tfrac{1}{\mu}) SW(\mcal)$, proving the claim. 
\end{proof}

Note that when $\alpha = 1$, the bounds in Propositions \ref{prop:bounds1} and \ref{prop:bounds2} coincide. 
Also note that by relaxing the stability requirement to $\alpha$-stability, the lower bound on the efficiency guarantee decreases whereas the upper bound on the efficiency guarantee increases. We refer to Figure~\ref{fig:stmin-plot} for an illustration of these bounds.

\section{Computing $\alpha$-Stable Matchings} \label{sec:computing-alpha-stable-matchings}

Our aim for this section is to design an algorithm which, given any instance of $\alpha$-SMCTI with $\mu$-bounded valuations, achieves an approximation guarantee of 1 whenever $\alpha\leq \nicefrac{\mu}{\mu+1}$, and of $\nicefrac{1}{\alpha}\cdot \nicefrac{\mu}{\mu+1}$ otherwise. Note that in light of Proposition \ref{prop:bounds2}, this is best possible. We start by introducing our algorithm, called \emph{$\alpha$-Boosting}, which builds upon \citep{gale1962college} and \citep{Colinibaldeschi2024Trust}.

\paragraph{$\alpha$-Boosting Algorithm.} Given a preference graph $\pg$ and a parameter $\alpha \in (0,1]$ as input, $\alpha$-Boosting first computes an optimal matching $\opt$ for $\pg$. If $\opt$ is $\alpha$-stable, then $\opt$ is returned. Otherwise, $\alpha$-Boosting defines an instance $\pga$ with the same underlying bipartite graph as $\pg$ and valuations $(v'_{ij}, w'_{ji})$ with 
\[
v'_{ij}=\begin{cases} 
      \frac{1}{\alpha} v_{ij}, & \text{if }(i,j)\in \opt, \\
      v_{ij}, & \text{otherwise,}
   \end{cases} 
\]
and $w'_{ji}$ defined analogously.
It then returns the matching $\mcal$ that is obtained by running the version of the Gale-Shapley algorithm that handles incomplete preferences (see \citep{Gusfield1989Stable}) on $\pga$, with ties broken arbitrarily.\footnote{The Gale–Shapley algorithm operates on ordinal preferences, but any set of cardinal preferences can be uniquely converted into this form.}
\medskip

Denote $m = |E(\pg)|$. 
Since an optimal matching for $\pg$ can be computed in time $O(m^{1+o(1)})$ \citep{Chen:2025}, 
$\pga$ can be constructed in time $O(m)$, and the Gale-Shapley algorithm on $\pga$ runs in time $O(m \log n)$, the complexity of $\alpha$-Boosting is $O(m \log n)$. 
In the special case of $\mu = 1$, an algorithm with runtime of $O(n^2)$ has already been proposed by \citet{anshelevich2013anarchy}. However, it is not trivial whether this algorithm can be generalized for $\mu \in (0,1)$. 

Note that $\mcal$ is a matching for $\pg$ if and only if it is a matching for $\pga$. Throughout this section, we will specify the preference graph ($\pg$ or $\pga$) in the subscript when relevant, and omit it otherwise. The main result of this section is the following theorem:

\begin{theorem} \label{th:boostingAlg}
    Let $\alpha \in (0,1]$. Then, for any instance of $\alpha$-SMCTI, $\alpha$-Boosting returns an $\alpha$-stable matching and achieves an approximation guarantee of $f(\alpha, \mu)$ with
    \begin{equation} \label{eq:guarantee-alpha-boosting}
     f(\alpha, \mu) = \begin{cases} 
      1, & \text{if } \alpha \le \frac{\mu}{\mu+1}, \\
      \frac{1}{\alpha}\cdot \frac{\mu}{\mu+1}, & \text{otherwise,}
      \end{cases}
    \end{equation}
    which is best possible. 
\end{theorem}
\begin{proof}
Consider an instance $\pg$ of $\alpha$-SMCTI. 
If $\alpha\leq \frac{\mu}{\mu+1}$, then $\alpha$-Boosting returns $\opt$ of $\pg$ which is $\alpha$-stable (Proposition \ref{prop:bounds1}), thus achieving an approximation guarantee of 1 in this case. 
Therefore, assume that $\alpha>\frac{\mu}{\mu+1}$ and let $\mcal$ be the matching returned by $\alpha$-Boosting.
We start by showing that $\mcal$ is $\alpha$-stable. Let $(i,j)$ be a pair which is not in $\mcal$. We distinguish two cases:
\begin{enumerate}
    \item If $i$ did not propose to $j$ during the execution of the Gale-Shapley algorithm, then we must have that $v'_{i\mcal(i)}\geq v'_{ij}$. By definition, it holds that $\frac{1}{\alpha}v_{i\mcal(i)}\geq v'_{i\mcal(i)}\geq v'_{ij}\geq v_{ij},$ so in particular, $v_{i\mcal(i)}\geq \alpha v_{ij}$. Therefore, $(i,j)$ cannot be $\alpha$-blocking.
    \item Otherwise, $i$ proposed to $j$ but $j$ rejected $i$, so it must hold that $w'_{j\mcal(j)}\geq w'_{ji}$. Again by definition, it holds that $\frac{1}{\alpha}w_{j\mcal(j)}\geq w'_{j\mcal(j)}\geq w'_{ji}\geq w_{ji},$ so $w_{j\mcal(i)}\geq \alpha w_{ji}$ and $(i,j)$ cannot be $\alpha$-blocking.
\end{enumerate}

Next, we establish the approximation guarantee by fixing an optimal matching $\opt$ for $\pg$ and constructing a mapping $g: \opt \to \mcal$. The main idea is that $g$ maps each edge in the optimal matching either to itself or to an adjacent edge of comparable value.

Let $(i,j) \in \opt$. If $(i,j) \in \mcal$, we set $g(i,j) = (i,j)$. Otherwise, $v'_{i\mcal(i)}\geq v'_{ij}$ or $w'_{j\mcal(j)}\geq w'_{ji}$. If the first condition holds, we set $g(i,j)=(i,\mcal(i))$. Since $(i,\mcal(i))\notin\mcal^*$, it follows that $v_{i\mcal(i)}=v'_{i\mcal(i)}\geq v'_{ij}=\frac{1}{\alpha}v_{ij}$, and combining with the definition of $\mu$-bounded valuations leads to
\[
\frac{v_{i\mcal(i)}+w_{\mcal(i)i}}{v_{ij}+w_{ji}} 
    \geq \frac{(1+\mu)v_{i\mcal(i)}}{( 1+\frac{1}{\mu} ) v_{ij}} 
    \geq \frac{\frac{1}{\alpha}(1+\mu)v_{ij}}{( 1+\frac{1}{\mu} ) v_{ij}} 
    = \frac{\mu}{\alpha}.
\]
If \textit{only} the second condition holds, we set $g(i,j)=(\mcal(j), j)$, which yields an analogous bound.
Now, if $(i,j) \in \mcal$ and $(i,j) \notin \opt$, it could be that two edges $(i,j') \in \opt$ and $(i',j) \in \opt$ are mapped to $(i,j)$. 
In this case we have that $\alpha v_{ij}\geq v_{ij'}$ and $\alpha w_{ji}\geq w_{ji'}$, so
\begin{align*}
    \frac{v_{ij}+w_{ji}}{v_{ij'}+w_{j'i}+v_{i'j}+w_{ji'}} &
    \geq \frac{\frac{1}{\alpha}v_{ij'}+\frac{1}{\alpha}w_{ji'}}{v_{ij'}+w_{j'i}+v_{i'j}+w_{ji'}} 
    \geq \frac{\frac{1}{\alpha}v_{ij'}+\frac{1}{\alpha}w_{ji'}}{( 1+\frac{1}{\mu}) ( v_{ij'}+w_{ji'})} 
    = \frac{1}{\alpha}\cdot \frac{\mu}{\mu+1}.
\end{align*}

Since $1>\frac{1}{\alpha}\cdot\frac{\mu}{\mu+1}$, $\frac{\mu}{\alpha}>\frac{1}{\alpha}\cdot\frac{\mu}{\mu+1}$ and edges in $\opt$ are mapped to adjacent edges in $M$ or itself (if in $M$), we obtain
\begin{align*}
\tfrac{1}{\alpha}\cdot \tfrac{\mu}{\mu+1}SW(\opt) 
&= \tfrac{1}{\alpha}\cdot\tfrac{\mu}{\mu+1}\sum_{(i,j)\in\opt}(v_{ij}+w_{ji})  \le \sum_{(i,j)\in\mcal}(v_{ij}+w_{ji}) = SW(\mcal),
\end{align*}
proving that $\alpha$-Boosting is $f(\alpha, \mu)$-approximate as in \eqref{eq:guarantee-alpha-boosting}. 
Furthermore, this is best possible by Proposition \ref{prop:bounds1}.
\end{proof}

\section{Hardness of Optimal $\alpha$-Stable Matchings} \label{sec:hardness}

As shown in the previous section, $\alpha$-Boosting always returns an $\alpha$-stable matching that recovers at least $\nicefrac{1}{\alpha}\cdot \nicefrac{\mu}{(\mu+1)}$ of the optimal social welfare. 
However, it need not be an optimal \emph{$\alpha$-stable} matching. 
In this section, we consider the problem of finding an optimal $\alpha$-stable matching.

\paragraph{Integer Linear Program of $\alpha$-SMCTI.} 

Given a matching $\mcal$ for a preference graph $\pg(L\cup R, (v_{ij}), (w_{ji}))$, we define its \iit{assignment matrix} as $x(\mcal)=(x_{ij})_{i\in L, j\in R}$, where $x_{ij}=1$ if the pair $(i,j)\in \mcal$, and $x_{ij}=0$ otherwise. 
The set of all matchings can be described by the following system of inequalities:
\begin{align}
    \textstyle\sum_{j\in R}x_{ij}&\leq 1 \qquad\;\;\; \forall i\in L \label{eq:constraint_one} \\
    \textstyle\sum_{i\in L}x_{ij}&\leq 1 \qquad\;\;\; \forall j\in R \label{eq:constraint_two}\\
    x_{ij}& \in \{0, 1\} \quad \forall (i,j)\in E(\pg).\label{eq:constraint_three}
\end{align}

Below, we show that the $\alpha$-stability property can be described by two different linear constraints.

\begin{proposition}\label{prop:LP_alpha_stab2}
    Let $\mcal$ be a matching for $\pg(L\cup R, (v_{ij}), (w_{ji}))$ and let $\alpha \in (0,1]$. Then, $\mcal$ is $\alpha$-stable if and only if  $x(\mcal)$ satisfies    \begin{equation}\label{eq:LP_alpha_stab2}
        \sum_{l:v_{il}\geq \alpha v_{ij}}x_{il}+\sum_{k:w_{jk}\geq \alpha w_{ji}}x_{kj}\geq 1 
        \quad \forall (i,j)\in E(\pg).
    \end{equation}
\end{proposition}
\begin{proof}
Towards a contradiction, suppose that $M$ is an $\alpha$-stable matching and there exists $(i,j) \in E(\pg)$ such that \eqref{eq:LP_alpha_stab2} does not hold. 
In this case it must be that both summations in \eqref{eq:LP_alpha_stab2} are equal to 0, and so $v_{i\mcal(i)}<\alpha v_{ij}$ and $w_{j\mcal(j)}<\alpha w_{ji}$. This implies that $(i,j)$ is an $\alpha$-blocking pair which contradicts that $M$ is $\alpha$-stable. 
For the other direction note that if \eqref{eq:LP_alpha_stab2} is satisfied for a matching $M$, this implies that there is no $\alpha$-blocking pair $(i,j)\in E(\pg)$ as $(i,j)$ would otherwise not satisfy \eqref{eq:LP_alpha_stab2}, implying that $M$ is $\alpha$-stable.
\end{proof}

\begin{proposition}\label{prop:LP_alpha_stab1}
    Let $\mcal$ be a matching for $\pg(L\cup R, (v_{ij}), (w_{ji}))$ and let $\alpha\in(0,1]$. Then, $\mcal$ is $\alpha$-stable if and only if $x(\mcal)$ satisfies
    \begin{equation}\label{eq:LP_alpha_stab1}
        \sum_{k:\alpha w_{ji}>w_{jk}}x_{kj}-\sum_{l:v_{il}\geq\alpha v_{ij}}x_{il} \leq 0 
        \quad \forall (i,j)\in E(\pg).
    \end{equation}
\end{proposition}
\begin{proof}
Let $\mcal$ be an $\alpha$-stable matching, and suppose that for a certain $(i,j)\in E(\pg)$ we have that $\sum_{k:\alpha w_{ji}>w_{jk}}x_{kj}=1$. This means that $j$ is paired with an agent they find less desirable than $i$ by a factor of $\alpha$. Since $\mcal$ is $\alpha$-stable, we must have $v_{i\mcal(i)}\geq \alpha v_{ij}$. In particular, $\sum_{l:v_{il}\geq\alpha v_{ij}}x_{il}=1$. This settles one direction. For the other direction, we proceed by contrapositive. Suppose that $\mcal$ is not $\alpha$-stable, and let $(i,j)$ be an $\alpha$-blocking pair, so $v_{i\mcal(i)}<\alpha v_{ij}$ and $w_{j\mcal(j)}<\alpha w_{ji}$. The first implies that $\mcal(i)\notin \{l:v_{il}\geq \alpha v_{ij}\}$, and therefore $\sum_{l:v_{il}\geq\alpha v_{ij}}x_{il}=0$. The second implies that $\mcal(j)\in \{k:\alpha w_{ji}>w_{jk}\}$, so $\sum_{k:\alpha w_{ji}>w_{jk}}x_{kj}=1$ and \eqref{eq:LP_alpha_stab1} does not hold for $(i,j)$ which concludes the proof.
\end{proof}

Using Propositions~\ref{prop:LP_alpha_stab2} and \ref{prop:LP_alpha_stab1}, we can describe the set of all $\alpha$-stable matchings via constraints 
(\ref{eq:constraint_one})--(\ref{eq:constraint_three}), together with either (\ref{eq:LP_alpha_stab2}) or (\ref{eq:LP_alpha_stab1}). 
Thus, computing an $\alpha$-stable matching is equivalent to solving the integer program that maximizes the objective $\sum_{(i,j) \in E(\pg)} (v_{ij}+ w_{ji}) x_{ij}$ subject to these constraints. The crux, however, is that integer programs are not known to be solvable in polynomial time. 

Notably, for the special case of stable matchings (i.e., $\alpha = 1$) with strict preference orders, \citet{Rothblum1992} and \citet{VandeVate1989} showed that the extreme points of the polytope defined by \eqref{eq:constraint_one}, \eqref{eq:constraint_two} and 
\begin{equation}
    x_{ij} \ge 0 \quad \forall (i,j)\in E(\pg),\label{eq:relaxation}
\end{equation}
together with 
(\ref{eq:LP_alpha_stab2}) or (\ref{eq:LP_alpha_stab1}), are integral.\footnote{In these papers, preference lists are given as order relations, but the generalization to cardinal preferences is almost immediate.} 
Consequently, in this case, solving the corresponding linear program (which can be done in polynomial time) yields an integral optimal solution. That is, for $\alpha = 1$ we can find an optimal $\alpha$-stable matching by solving the respective linear program. 
This naturally raises the question of whether the same holds for $\alpha < 1$. We answer this question in the negative by showing that the problem is NP-hard, even when $\alpha$ is arbitrarily close to 1.

To prove our hardness results, we consider certain special cases of our $\alpha$-SMCTI problem. 
We denote these more restrictive variants by omitting the respective letters from the acronym.
In particular, SMCT (SMT) refers to the stable matching variant with cardinal (ordinal) preferences, ties allowed, and complete preference lists. 
Given a weak order relation $\succeq_i$ for $i \in L$, we write $j_1 \succeq_i j_2$ to indicate that $i$ (weakly) prefers $j_1$ over $j_2$; the relation $\succeq_j$ for $j \in R$ is defined analogously. 

Given an instance $I$ of SMT and $(i,j)\in L\times R$, we define the \iit{cost} $c_{ij}$ of agent $i$ for $j$ as the position $j$ occupies in $i$'s preference list (the cost $c_{ji}$ of agent $j$ for $i$ is defined analogously). Agents in the same tie group share the same cost, which will be equal to the number of agents ranked strictly above them plus one. We define the \iit{weight} of $\mcal$ as $w(\mcal):=\sum_{u\in L\cup R}c_{u\mcal(u)}$, where $c_{u\mcal(u)}=\infty$ whenever $u$ is lonely under $\mcal$. 
We say that a matching is \iit{egalitarian} if its weight is minimum among all stable matchings of $I$.

Let {\textsc{Egalitarian-SMT-OPT}} denote the optimization problem of finding an egalitarian stable matching given an instance of SMT. 
\citet{Manlove2002} proved that \textsc{Egalitarian-SMT-OPT} is hard to approximate. In particular, this means that the following decision problem is also NP-hard:
\begin{decisionproblem}
    {\textsc{Min-Cost-SMT}}
    {Given an instance of SMT and $K\in \mathbb{R}^+$,}
    { is there a stable matching $\mcal$ such that $w(\mcal)\leq K$?}
\end{decisionproblem}

The decision problem we are interested in here is: 

\begin{decisionproblem}
    {\textsc{Max-Weight}-$\alpha$-SMCT}
    {Given an instance of $\alpha$-SMCT and $\omega\in \mathbb{R}^+$,}
    { is there an $\alpha$-stable matching $\mcal$ such that $SW(\mcal) \geq \omega$?}
\end{decisionproblem}

\begin{theorem}\label{thm:maxweight_NPhard}
Given $\alpha\in (\frac{n-1}{n},1)$, \textsc{Max-Weight-$\alpha$-SMCT} is NP-complete.
\end{theorem}

\begin{proof}
It can easily be verified that \textsc{Max-Weight}-$\alpha$-SMCT is in NP. 
Let $\alpha\in \lbr\frac{n-1}{n}, 1\rbr$. 
We define a polynomial-time reduction from SMT to $\alpha$-SMCT.

Given an instance $I$ of SMT, let $P_i = \{j_1^i, \ldots, j_n^i\}$ denote the preference list of $i$, which potentially includes ties. 
Let $\pg$ be an instance of $\alpha$-SMCT with the same underlying graph as $I$. For each $(i, j)\in L\times R$, we define the cardinal preferences as $v_{ij} = -c_{ij}+n+1$ and $w_{ji} = -c_{ji}+n+1$. Note that as the costs are in the range $[1,n]$, so are the valuations by construction.
    
It trivially holds that $\mcal$ is a matching for $\pg$ if and only if $\mcal$ is a matching for $I$.
We will prove that a matching $\mcal$ is $\alpha$-stable for $\pg$ if and only if $\mcal$ is stable for $I$ by showing that a pair $(i,j)$ is $\alpha$-blocking for $\pg$ if and only if it is blocking for $I$. 
Note that one direction holds by construction: take $(i,j)$ to be $\alpha$-blocking for $\pg$, i.e. $v_{i\mcal(i)} < \alpha v_{ij} < v_{ij}$ and $w_{j\mcal(j)} < \alpha w_{ji} < w_{ji}$. This can only be the case if $c_{ij} < c_{i\mcal(i)}$ and $c_{ji} < c_{j\mcal(j)}$, so $(i,j)$ is blocking for $I$. 
For the other direction, suppose that $(i,j)$ is blocking for $I$, i.e. $c_{ij}<c_{i\mcal(i)}$ and $c_{ji}<c_{j\mcal(j)}$. This implies that $\mcal(i)$ and $j$ belong to different tie groups in $i$'s preference list, and $i$ and $\mcal(j)$ belong to different tie groups in $j$'s. 
We want to show that $v_{i\mcal(i)}<\alpha v_{ij}$ and $w_{j\mcal(j)}<\alpha w_{ji}$, thus it is enough to show that the valuation of an element in a tie group times $\alpha$ is greater than the valuation of any element in a following tie group. 
The most restrictive case happens when each tie group contains exactly one agent, in which case the largest ratio arises between the first and second agent in the preference list. More formally, let $2\le x \le n$, then $\frac{v_{i\mcal(i)}}{v_{ij}} \leq \frac{x-1}{x} \leq \frac{n-1}{n} < \alpha$, which indeed gives $v_{i\mcal(i)}<\alpha v_{ij}$. The case of $w_{j\mcal(j)}<\alpha w_{ji}$ follows analogously.

Take $K\in \lbr 1, 2n(n+1)\rbr$, and set $\omega=2n(n+1)-K$. 
We now show that there exists an $\alpha$-stable matching $\mcal$ for $\pg$ with $SW(\mcal)\geq \omega$ if and only if there exists a stable matching $\mcal'$ for $I$ with $w(\mcal')\leq K$. 
Suppose that $\mcal$ is an $\alpha$-stable matching for $\pg$ with $SW(\mcal)\geq \omega$. 
It holds that
\begin{align*}
SW(\mcal) &= \! \!  \sum_{(i,j)\in \mcal} \! \! (v_{ij}+w_{ji}) = \! \! \sum_{(i,j)\in \mcal} \! \! (-c_{ij}-c_{ji}+2n+2)
=-w(\mcal)+2n(n+1),
\end{align*}
where the last equality is due to $\mcal$ being perfect (since preference orders are complete). Thus, $w(\mcal)\leq 2n(n+1)-\omega=K$. 
In particular, $\mcal$ is a stable matching for $I$ with $w(\mcal)\leq K$. 
Analogously, suppose that $\mcal$ is a stable matching for $I$ with $w(\mcal) \le K$. 
It holds that
\begin{align*}
w(\mcal) &= \! \! \sum_{(i,j)\in \mcal} \! \! (c_{ij} + c_{ji}) = \! \! \sum_{(i,j)\in \mcal} \! \!  (-v_{ij}-w_{ji}+2n+2)
= - SW(\mcal) +2n(n+1),
\end{align*}
and so $SW(\mcal) \geq 2n(n+1) - K = \omega$. Thus $\mcal$ is an $\alpha$-stable matching for $\pg$ with $SW(\mcal) \geq \omega$. 
\end{proof}

Note that we prove that \textsc{Max-Weight-$\alpha$-SMCT} is NP-complete for $\alpha \in (\frac{n-1}{n},1)$. 
As mentioned above, the problem is solvable in polynomial time when $\alpha = 1$, and when $\alpha \le \frac{\mu}{\mu+1}$ by Theorem \ref{th:boostingAlg}. This reveals an interesting threshold phenomenon. 
Note that $\mu$ is determined by the valuations defined in the reduction and can be as small as $\nicefrac{1}{n}$; the proof can be adapted to larger values of $\mu < 1$ by altering the way the valuations are defined.
Clearly, Theorem~\ref{thm:maxweight_NPhard} implies NP-hardness for more general optimization versions such as $\alpha$-SMCTI.

\section{Conclusion and Future Work}

We study the interplay between stability and social welfare in two-sided matching markets with cardinal preferences, ties, and incomplete preference lists. By adopting the relaxed notion of $\alpha$-stability, we provide a complete characterization of the stability--efficiency tradeoff under asymmetric valuations and show that relaxing stability can substantially improve achievable efficiency guarantees. We derive a polynomial-time algorithm that computes an $\alpha$-stable matching attaining the best possible efficiency guarantee. Finally, our hardness results show that computing optimal $\alpha$-stable matchings becomes intractable even under slight relaxations of stability, i.e., for $\alpha < 1$.

Our work leaves several interesting questions for future research.
A natural direction is to further refine the complexity landscape of $\alpha$-stable matchings. In particular, it would be interesting to determine whether computing an optimal $\alpha$-stable matching is NP-hard for all $\alpha \in \bigl(\nicefrac{\mu}{\mu+1}, 1\bigr)$, and to understand the approximability of optimal $\alpha$-stable matchings within this range. Depending on these hardness results, a natural next step is to leverage predictions to guide the computation of near-optimal $\alpha$-stable matchings; notably, our algorithm already draws inspiration from this paradigm. Other promising directions include extending the framework to many-to-one matching markets and exploring hybrid notions of stability that incorporate different behavioral assumptions.


\bibliography{references}

\end{document}